\documentclass[a4paper,10pt]{amsart}

\usepackage[latin1]{inputenc}
\usepackage[english]{babel}
\usepackage{fontenc}
\usepackage{graphicx}
\usepackage{amsmath}
\usepackage{amsfonts}
\usepackage{amssymb}
\usepackage{amsthm}
\usepackage{appendix}

\newtheorem{theorem}{Theorem}[section]
\newtheorem{definition}[theorem]{Definition}
\newtheorem{proposition}[theorem]{Proposition}

\newtheorem{corollary}[theorem]{Corollary}

\theoremstyle{remark}
\newtheorem{remark}[theorem]{Remark}
\newtheorem{example}[theorem]{Example}

\newcommand{\NN}{\mathbb{N}}
\newcommand{\CC}{\mathbb{C}}

\newcommand{\cC}{\mathcal{C}}
\newcommand{\cD}{\mathcal{D}}
\newcommand{\cB}{\mathcal{B}}
\newcommand{\cR}{\mathcal{R}}

\newcommand{\R}{\mathbb{R}}
\newcommand{\C}{\mathbb{C}}

\renewcommand{\Im}{\operatorname{Im}}

\newcommand{\D}{\mathbb{D}}
\newcommand{\B}{\mathcal{B}}
\renewcommand{\O}{\mathcal{O}}
\newcommand{\F}{\mathcal{F}}

\renewcommand{\epsilon}{\varepsilon}

\numberwithin{equation}{section}

\title[Free Deterministic Equivalents]{Free Deterministic Equivalents, Rectangular Random Matrix Models, and Operator-Valued Free Probability Theory} 

\author[R. Speicher and C. Vargas (Appendix by T. Mai)]{Roland Speicher and Carlos Vargas\\
With an Appendix by Tobias Mai}

\address{Saarland University, Fachbereich Mathematik, Postfach 151150,
66041 Saarbr\"ucken, Germany}

\begin{document}

\maketitle

\begin{abstract}
Motivated by the asymptotic collective behavior of random and deterministic matrices, we propose an approximation (called ``free deterministic equivalent'') to quite general random matrix models, by replacing the matrices with operators satisfying certain freeness relations. We comment on the relation between our free deterministic equivalent and deterministic equivalents considered in the engineering literature.
We do not only consider the case of square matrices, but also show how rectangular matrices can be treated. Furthermore, we emphasize how operator-valued free probability techniques can be used to solve our free deterministic equivalents.

As an illustration of our methods we show how the free deterministic equivalent of a random matrix model from \cite{CouDeHo}
can be treated and we thus recover in a conceptual way the results from \cite{CouDeHo}.

On a technical level,
we generalize a result from scalar valued free probability, by showing that randomly rotated deterministic matrices of different sizes are asymptotically free from deterministic rectangular matrices, with amalgamation over a certain algebra of projections.

In the Appendix, we show how estimates for differences between Cauchy transforms can be extended from a neighborhood of infinity to a region close to the real axis. This is of some relevance if one wants to
compare the original random matrix problem with its free deterministic equivalent.
\end{abstract}

\section{Introduction}

Voiculescu's results \cite{V-Inventiones, VDN, HP, NiSp} on the asymptotic freeness of random matrices have provided a way to understand the collective asymptotic behavior of a variety of random matrix ensembles. Many known results from the random matrix literature have found straight-forward solutions via free probability. Additionally, this perspective has lead to new results (e.g. the asymptotic freeness of randomly rotated deterministic ensembles). 
The engineering community, in particular, has added concepts and results from free probability theory to their toolbox for dealing with random matrix ensembles in the context of wireless communications, see, e.g., \cite{Tulino-Verdu}. However, still there is the widespread perception that free probability cannot deal with situations where rectangular matrices are involved or where one does not have an asymptotic eigenvalue distribution of some of the involved ensembles. In the latter context, so-called ``deterministic equivalents'' (apparently going back to Girko \cite{Gir}, see also \cite{HLN}) have recently been studied quite extensively in the engineering literature.

In this paper we will show 
\begin{itemize}
 \item 
how those deterministic equivalents can be understood and formalized in a conceptual way by using free probability theory
\item
and that this can actually also be done in the presence of rectangular random matrices  
\end{itemize}

Operator-valued free probability \cite{V-operatorvalued, Sp-Memoir} will play a crucial role in these investigations. Firstly, it is actually needed to deal in a nice way with rectangular random matrices - following the ideas of Benaych-Georges \cite{Be2} we will embed the rectangular matrices in a large matrix space, which has the structure of a projection-valued probability space (which is one of the simplest cases of operator-valued spaces). Secondly, for more advanced random matrix models one usually has to rely on operator-valued freeness results to derive the equations for the deterministic equivalents. The fundamental observation that 
the operator valued version of free probability allows to understand a broader class of matrix ensembles, such as those consisting of band matrices or block matrices, goes back to Shlyakhtenko \cite{S-band}. See also \cite{ROBS} for a first use of operator-valued freeness in the context of wireless communication problems.

In order to illuminate, and also to show the power of our general approach we will re-derive the equations for the deterministic equivalent of the Cauchy transform of a model considered in \cite{CouDeHo} (in the asymptotic regime this model was also considered before in \cite{PCH}). The model is the following:
\begin{equation}
\Phi_{N}=\sum_{i=1}^{k}H_{i}^{\left(N\right)}U_{i}^{\left(N\right)}T_{i}^{\left(N\right)}U_{i}^{\left(N\right)*}H_{i}^{\left(N\right)*},\label{mimomodel}
\end{equation}
where $k$ is a fixed integer and
\begin{itemize}
 \item[(A1)] 
For each $1\leq i\leq k$, $(N_{i}^{\left(N\right)})_{N\in\mathbb{N}}$ is a sequence of natural numbers such that $N_{i}^{\left(N\right)}/N\to c_{i}\in\left(0,\infty\right)$.
\item[(A2)] 
For each $1\leq i\leq k$, $H_{i}^{\left(N\right)}$ is an $N\times N_{i}^{\left(N\right)}$ deterministic matrix and $T_{i}^{\left(N\right)}$ is an $N_{i}^{\left(N\right)}\times N_{i}^{\left(N\right)}$ deterministic matrix.
\item[(A3)] 
The matrices $U_{i}^{\left(N\right)}$ are independent Haar-distributed unitaries from $\mathcal{U}(N_{i}^{\left(N\right)})$.
\end{itemize}

The paper is organized as follows: 

Section 2 is devoted to the basic theory on operator-valued free probability. We state some results from \cite{Be2}, generalizing the asymptotic freeness results by Voiculescu to the case where we allow rectangular matrices of different sizes. These results are stated in the framework of rectangular probability spaces. Our contribution to the theory is a generalization of the asymptotic freeness of randomly rotated deterministic matrices.

In Section 3 we define the notion of a ``free deterministic equivalent'' for a random matrix model, based on our understanding of the limiting collective behavior of random and deterministic rectangular matrices provided in the previous section. We will also point out the relation between our free deterministic equivalent and the deterministic equivalents from the engineering literature: whereas the latter is an approximation on the level of equations for the Cauchy transforms, our approximation is directly on the level of operators;
but in the end they coincide, i.e. the Cauchy transform of our free deterministic equivalent satisfies exactly the deterministic equivalent equations. 

In Section 4 we present the free deterministic equivalent of the random matrix model \eqref{mimomodel}
and derive, by using operator-valued free probability results, the equations for its Cauchy transform. This will recover the results from \cite{CouDeHo}.

In the last section we give the proof of our result on randomly rotated rectangular matrices from Section 2.

In the Appendix (written by T. Mai) we show how estimates for differences between Cauchy transforms can be extended from a neighborhood of infinity to a region close to the real axis. This is of some relevance if one wants to
compare the original problem with its free deterministic equivalent.

\section{Operator-Valued Free Probability}

\subsection{Preliminaries}

We begin by recalling the basic concepts on operator-valued free probability, for more details see \cite{V-operatorvalued, Sp-Memoir}.

Let $\mathcal{B}$ be a unital algebra. A $\mathcal{B}$\textit{-probability space} is a pair $\left(\mathcal{A},\mathbf{F}\right)$ consisting of an algebra $\mathcal{A}$ containing $\mathcal{B}$ and a linear map  $\mathbf{F}:\mathcal{A}\to\mathcal{B}$ satisfying
\begin{equation*}
\mathbf{F}\left(bab'\right) = b\mathbf{F}(a)b', \qquad \forall b,b'\in\mathcal{B},a\in\mathcal{A}
$$
and
$$
\mathbf{F}\left(1\right) = 1. 
\end{equation*}
Such a map $\mathbf{F}$ is called a \textit{conditional expectation}. From the properties above we observe that  $\mathbf{F}\left(b\right)=b$ for all $b\in\mathcal{B}$.

A collection of subalgebras $\mathcal{B}\subset A_{1},\dots,A_{n}\subset\mathcal{A}$ is \textit{free with amalgamation over} $\mathcal{B}$ (or, $\mathcal{B}$\textit{-free}, for short) if we have that for every $k\in\mathbb{N}$
\begin{equation*}
\mathbf{F}\left(a_{i\left(1\right)}\dots a_{i\left(k\right)}\right)=0 
\end{equation*}
whenever we have $\mathbf{F}\left(a_{i\left(j\right)}\right)=0$ and 
$a_{i(j)}\in A_{\varepsilon\left(j\right)}$ for some $\varepsilon(j)\in\{1,\dots,n\}$, $1\leq j\leq k$ and that $\varepsilon\left(1\right)\neq\varepsilon\left(2\right)\neq\dots\neq
\varepsilon\left(k\right)$. (Note that $\varepsilon\left(1\right)=\varepsilon\left(3\right)$ is allowed.)

Elements $a_{1},\dots,a_{n}$ are $\mathcal{B}$\textit{-free} if the algebras $\left\langle a_{1},\mathcal{B}\right\rangle ,\dots,\left\langle a_{n},\mathcal{B}\right\rangle $ are $\mathcal{B}$-free. ($\left\langle a_{1},\mathcal{B}\right\rangle$ is here the subalgebra of $\mathcal{A}$ which is generated by $a_1$ and by $\mathcal{B}$.)

For $n\in\mathbb{N}$, a $\mathbb{C}$-multilinear map $f:\mathcal{A}^{n}\to\mathcal{B}$ is called $\mathcal{B}$\textit{-balanced} if it satisfies the $\mathcal{B}$-bilinearity conditions, that for all $b,b'\in\mathcal{B}$, $a_{1},\dots,a_{n}\in\mathcal{A}$, and for all $r=1,\dots,n-1$
\begin{eqnarray*}
f\left(ba_{1},\dots,a_{n}b'\right) &=& bf\left(a_{1},\dots,a_{n}\right)b'\\
f\left(a_{1},\dots,a_{r}b,a_{r+1},\dots,a_{n}\right) &=& f\left(a_{1},\dots,a_{r},ba_{r+1}\dots,a_{n}\right)
\end{eqnarray*}

For a finite totally ordered set $S$, a \textit{partition} $\pi=\{V_1,\dots,V_k\}$ is a decomposition of $S$ into pairwise disjoint non-empty subsets, called \textit{blocks}. A partition is \textit{non-crossing} if whenever we have $a<b<c<d\in S$ such that $a,c \in V_i$ and $b,d \in V_j$, then $V_i=V_j$. We denote by $NC(n)$ the set of non-crossing partitions of $[n]:=\{1,\dots ,n\}$ and $NC:=\bigcup_{n \geq 0}NC(n)$. The partial order of reverse refinement ($\sigma \leq \pi$ iff every block of $\sigma$ is completely contained in a block of $\pi$) turns $NC(n)$ into a lattice. We denote the smallest and largest partition of $NC(n)$ by $\mathbf{0}_n$ and $\mathbf{1}_n$, respectively.

A collection of $\mathcal{B}$-balanced maps $\left(f_{\pi}\right)_{\pi\in NC}$ is said to be \textit{multiplicative} (w. r. t. the lattice of non-crossing partitions) if, for every $\pi\in NC$, $f_{\pi}$ is computed using the block structure of $\pi$ in the following way:

\begin{itemize}
 \item If $\pi=\mathbf{1}_{n}\in NC\left(n\right)$, we just write $f_{n}:=f_{\pi}$.

 \item If $\mathbf{1}_{n}\neq\pi=\left\{ V_{1},\dots,V_{k}\right\} \in NC\left(n\right),$ then by a known characterization of $NC$, there exists a block $V_{r}=\left\{ s+1,\dots,s+l\right\} $ containing consecutive elements. For any such a block we must have
\begin{equation*}
f_{\pi}\left(a_{1},\dots,a_{n}\right)=f_{\pi\backslash V_{r}}\left(a_{1},\dots,a_{s}f_{l}\left(a_{s+1},\dots,a_{s+l}\right),a_{s+l+1},\dots,a_{n}\right), 
\end{equation*}
where $\pi\backslash V_{r}\in NC\left(n-l\right)$ is the partition obtained from removing the block $V_{r}$.

\end{itemize}

We observe that a multiplicative family $\left(f_{\pi}\right)_{\pi\in NC}$ is entirely determined by $\left(f_{n}\right)_{n\in\mathbb{N}}$. On the other hand, every collection $\left(f_{n}\right)_{n\in\mathbb{N}}$ of $\mathcal{B}$-balanced maps can be extended uniquely to a multiplicative family $\left(f_{\pi}\right)_{\pi\in NC}$.

The \textit{operator-valued cumulants} $\left(\kappa_{\pi}\right)_{\pi\in NC}$ are indirectly and inductively defined as the unique multiplicative family of $\mathcal{B}$-balanced maps satisfying the (operator-valued) moment-cumulant formulas
\begin{equation*}
\mathbf{F}\left(a_{1}\dots a_{n}\right)=\sum_{\pi\in NC\left(n\right)}\kappa_{\pi}\left(a_{1},\dots,a_{n}\right) 
\end{equation*}

The conditional expectation yields also a multiplicative family by simply defining $\mathbf{F}_{n}:\mathcal{A}^{n}\to\mathcal{B}$ by
\begin{equation*}
\mathbf{F}_{n}\left(a_{1},\dots,a_{n}\right)=\mathbf{F}\left(a_{1}\dots a_{n}\right), 
\end{equation*}
and extending multiplicatively to $NC$. Actually, the relation between cumulants and moments is much richer. The cumulants can also be obtained from the moments via a M\"{o}bius inversion in the lattice of non-crossing partitions, and one has, for example
\begin{equation*}
\kappa_{\pi}\left(a_{1},\dots,a_{n}\right)=\sum_{\sigma\leq\pi}\mathbf{F}_{\sigma}\left(a_{1},\dots,a_{n}\right)\mu_n[\sigma,\pi], 
\end{equation*}
where $\mu_n:NC\left(n\right)^{2}\to\mathbb{C}$ is the \textit{M\"{o}bius function} in $NC\left(n\right)$.

By the \textit{cumulants of a tuple} $a_{1},\dots,a_{k}\in\mathcal{A}$, we mean the collection of all cumulant maps
\begin{equation*}
\begin{array}{cccc}
\kappa_{i_{1},\dots,i_{n}}^{a_{1},\dots,a_{k}}: & \mathcal{B}^{n-1} & \to & \mathcal{B},\\
 & \left(b_{1},\dots,b_{n-1}\right) & \mapsto & \kappa_{n}^{\mathcal{B}}\left(a_{i_{1}}b_{1},a_{i_2}b_2,\dots,a_{i_n}\right)\end{array} 
\end{equation*}
for $n\in\mathbb{N}$, $1\leq i_{1},\dots,i_{n}\leq k$. We sometimes write $\kappa_{i_{1},\dots,i_{n}}^{\mathcal{B};a_{1},\dots,a_{k}}$ to emphasize the underlying subalgebra.

A cumulant map $\kappa_{i_{1},\dots,i_{n}}^{a_{1},\dots,a_{k}}$ is \textit{mixed} if there exists $r<n$ such that $i_{r}\ne i_{r+1}$.

The main feature of the operator-valued cumulants is that they characterize freeness with amalgamation \cite{Sp-Memoir}: The random variables $a_{1},\dots,a_{n}$ are $\mathcal{B}$-free iff all their mixed cumulants vanish. 

The \textit{operator-valued Cauchy transform} is given by
\begin{equation*}
G_{a}^{\mathcal{B}}\left(b\right)=\mathbf{F}\left(\frac{1}{b-a}\right)=\sum_{n\geq0}\mathbf{F}\left(b^{-1}\left(ab^{-1}\right)^{n}\right) 
\end{equation*}

The \textit{operator-valued $\mathcal{R}$-transform} is given by
\begin{equation*}
\mathcal{R}_{a}^{\mathcal{B}}\left(b\right)=\sum_{n\geq1}\kappa_{n}^{\mathcal{B}}\left(a,ba,\dots,ba\right). 
\end{equation*}
The vanishing of mixed moments for free variables implies the additivity of the cumulants, and thus also the additivity of the $\cR$-transforms \cite{V-operatorvalued}: If $a_1$ and $a_2$ are $\cB$-free then we have for $b\in \cB$ that 
$\cR_{a_1+a_2}(b)=\cR_{a_1}(b)+\cR_{a_2}(b)$.

As in the scalar case, these transforms satisfy the functional equation
\begin{equation*}
G_{a}^{\mathcal{B}}\left(b\right)=\left(\mathcal{R}_{a}^{\mathcal{B}}\left(G_{a}^{\mathcal{B}}\left(b\right)\right)-b\right)^{-1} 
\end{equation*}

The above transforms and the functional equation are on one hand just identities for formal power series, but one can also consider them as relations for analytic functions, on appropriate domains. For more on the analytic properties in the operator-valued context, see \cite{V-operatorvalued,BePoVi}.

\subsection{Working on Different Levels}

One of the most fascinating aspects about the operator-valued version of free probability is that one can consider several conditional expectations on a single space, and compare the corresponding distributions.

In order to guarantee the existence of nice conditional expectations onto subalgebras one needs some analytic structure. Particular nice is the case of a tracial $W^*$-probability space $\left(\mathcal{A},\tau\right)$, which means that 
$\mathcal{A}$ is a $W^*$-algebra, or von Neumann algebra, (i.e., a unital subalgebra of bounded operators on a Hilbert space, closed in the weak operator topology), and that $\tau$ is a normal state which satisfies the trace condition $\tau(a_1a_2)=\tau(a_2a_1)$ for all $a_1,a_2\in \mathcal{A}$. We will only
encounter situations of this type. In particular, matrices (equipped with the
usual trace as state) and their limits are of this type.

If we start from a tracial $W^{*}$-probability space $\left(\mathcal{A},\tau\right)$, then every $W^{*}$-subalgebra $\mathbb{C}\subset\mathcal{B}\subset\mathcal{A}$ induces a unique conditional expectation $\mathbf{F}:\mathcal{A}\to\mathcal{B}$ compatible with $\tau$ in the sense that $\tau=\tau\circ\mathbf{F}$.

This means in particular that the scalar-valued Cauchy transform can be obtained by the projection of the operator-valued one, namely

\begin{equation*}
\tau\left(G_{a}^{\mathcal{B}}\left(z\mathbf{1}_{\mathcal{B}}\right)\right)=G_{a}^{\mathbb{C}}\left(z\right). 
\end{equation*}

When dealing with elements $a_{1},\dots,a_{n}$ in $\mathcal{A}$ which are not free, very often we can ``enlarge'' the algebra of scalars $\mathbb{C}$ to a subalgebra  $\mathbb{C}\subset\mathcal{B}\subset\mathcal{A}$, in such a way that $a_{1},\dots,a_{n}$ become $\mathcal{B}$-free. Then, a variety of results can be used to treat and simplify the $\mathcal{B}$-distribution of products and sums of these free elements. The scalar-valued distribution can then be obtained in the end by projecting with $\tau$.

We recall two results from \cite{NiShSp}. The first proposition gives conditions for (operator-valued) cumulants to be restrictions of cumulants with respect to a larger algebra.

\begin{proposition}

\label{thm:Restrictions}

Let $1\in\mathcal{D}\subset\mathcal{B}\subset\mathcal{A}$ be algebras such that $\left(\mathcal{A},\mathbf{E}\right)$ and $\left(\mathcal{B},\mathbf{F}\right)$ are respectively $\mathcal{B}$-valued and $\mathcal{D}$-valued probability spaces (and therefore $\left(\mathcal{A},\mathbf{F}\circ\mathbf{E}\right)$ is a $\mathcal{D}$-valued probability space) and let $a_{1},\dots,a_{k}\in\mathcal{A}$.

Assume that the $\mathcal{B}$-cumulants of $a_{1},\dots,a_{k}\in\mathcal{A}$ satisfy
\begin{equation*}
\kappa_{i_{1},\dots,i_{n}}^{\mathcal{B};a_{1},\dots,a_{k}}\left(d_{1},\dots,d_{n-1}\right)\in\mathcal{D}, 
\end{equation*}
for all $n\in\mathbb{N}$, $1\leq i_{1},\dots,i_{n}\leq k$, $d_{1},\dots,d_{n-1}\in\mathcal{D}$.

Then the $\mathcal{D}$-cumulants of $a_{1},\dots,a_{k}$ are just the restrictions of the $\mathcal{B}$-cumulants of $a_{1},\dots,a_{k}$, namely
\begin{equation*}
\kappa_{i_{1},\dots,i_{n}}^{\mathcal{B};a_{1},\dots,a_{k}}\left(d_{1},\dots,d_{n-1}\right)=\kappa_{i_{1},\dots,i_{n}}^{\mathcal{D};a_{1},\dots,a_{k}}\left(d_{1},\dots,d_{n-1}\right), 
\end{equation*}
for all $n\in\mathbb{N}$, $1\leq i_{1},\dots,i_{n}\leq k$, $d_{1},\dots,d_{n-1}\in\mathcal{D}$.

\end{proposition}

The second proposition gives a characterization of operator-valued freeness by the agreement of operator-valued cumulants in different levels.

\begin{proposition}

\label{thmNiShSp}

Let $1\in\mathcal{D}\subset\mathcal{B},\mathcal{N}\subset\mathcal{A}$ be algebras such that $\left(\mathcal{A},\mathbf{E}\right)$ and $\left(\mathcal{B},\mathbf{F}\right)$ are respectively $\mathcal{B}$-valued and $\mathcal{D}$-valued probability spaces. Then the first of the following statements implies the second.
\begin{itemize}
 \item[i)]
The subalgebras $\mathcal{D},\mathcal{N}$ are free with amalgamation over $\mathcal{B}$.
\item[ii)]
For every $k\in\mathbb{N}$, $n_{1},\dots,n_{k}\in\mathcal{N}$, $b_{1},\dots,b_{k-1}\in\mathcal{B}$, we have 
\begin{equation*}
\kappa_{n}^{\mathcal{B}}\left(n_{1}b_{1},\dots,n_{k-1}b_{k-1},n_{k}\right)=\kappa_{n}^{\mathcal{D}}\left(n_{1}\mathbf{F}\left(b_{1}\right),\dots,n_{k-1}\mathbf{F}\left(b_{k-1}\right),n_{k}\right). 
\end{equation*}
\end{itemize}
Moreover, the two statements become equivalent if we add the faithfulness condition on $\mathbf{F}:\mathcal{B}\to\mathcal{D}$, that if $\mathbf{F}\left(b_{1}b_{2}\right)=0$ for all $b_{2}\in\mathcal{B}$, then $b_{1}=0$.

\end{proposition}

\subsection{Rectangular Spaces}

Let $\left(\mathcal{A},\tau\right)$ be a tracial $W^{*}$-probability space endowed with pairwise orthogonal, non-trivial projections $p_{1},\dots,p_{k}$ adding up to one. Let $\mathcal{D}:=\left\langle p_{1},\dots,p_{k}\right\rangle $ denote the $W^{*}$-algebra generated by $\left\{ p_{1},\dots,p_{k}\right\} $. Then there exists a unique conditional expectation $\mathbf{F}:\mathcal{A}\to\mathcal{D}$ such that $\tau\circ\mathbf{F}=\tau $. We actually have that
\begin{equation}
\mathbf{F}\left(a\right)=\sum_{i=1}^{k}p_{i}\tau\left(p_{i}\right)^{-1}\tau\left(p_{i}a\right).\label{eq:condexp}
\end{equation}
With this, $\left(\mathcal{A},\mathbf{F}\right)$ becomes a $\mathcal{D}$-valued probability space.

These kind of projection-valued spaces are called \textit{rectangular probability spaces}. We will denote by $\mathcal{A}^{\left(i,j\right)}$ the set of elements  $a\in\mathcal{A}$ such that $a=p_{i}ap_{j}$. Elements in $\bigcup_{1\leq i,j\leq k}\mathcal{A}^{\left(i,j\right)}$ are called \textit{simple} and we write $\mathcal{A}^{\left(i\right)}:=\mathcal{A}^{\left(i,i\right)}$. Very often we will be interested in the compressed spaces $(\mathcal{A}^{\left(i\right)},\tau^{\left(i\right)})$, where $\tau^{\left(i\right)}\left(a\right)=\tau(p_{i})^{-1}\tau\left(a\right),$ for  $a\in\mathcal{A}^{\left(i\right)}$.

These spaces were introduced by Benaych-Georges in \cite{Be1, Be2}. They allow to treat rectangular matrices of different sizes in a single space. 

From now on, we will restrict the use of the caligraphic letter $\mathcal{D}$ to the $W^{*}$-algebra generated by the projections of a rectangular probability space. Therefore, $\mathcal{D}$-freeness will always mean freeness with amalgamation over the algebra of projections which defines the given rectangular probability space. It will be useful to specialize the corresponding notion of asymptotic freeness as well.

\begin{definition}

Let $\left(\mathcal{A},\tau\right)$ be a $(p_{1},\dots,p_{k})$-rectangular space and let $(\mathcal{A}_{n},\tau_{n})_{n\geq1}$ be a sequence of $(p_{1}^{\left(n\right)},\dots,p_{k}^{\left(n\right)})$-rectangular spaces. Let $a_{1},\dots,a_{m}\in\mathcal{A}$ and $a_{1}^{\left(n\right)},\dots,a_{m}^{\left(n\right)}\in\mathcal{A}_{n}$ be collections of simple elements. We say that  $(a_{1}^{\left(n\right)},\dots,a_{m}^{\left(n\right)})$ converges in $\mathcal{D}$-distribution to $(a_{1},\dots,a_{m})$ if  $(a_{1}^{\left(n\right)},\dots,a_{m}^{\left(n\right)},p_{1}^{\left(n\right)},\dots,p_{k}^{\left(n\right)})$ converges in $*$-distribution to  $(a_{1},\dots,a_{m},p_{1},\dots,p_{k})$, and we write
\begin{equation*} 
\left(a_{1}^{\left(n\right)},\dots,a_{m}^{\left(n\right)}\right)\overset{\mathcal{D}}\to (a_{1},\dots,a_{m})\qquad \text{as $n\to\infty$}.
\end{equation*}
If $a_{1},\dots,a_{m}$ are $(p_{1},\dots,p_{k})$-free, we say that $a_{1}^{\left(n\right)},\dots,a_{m}^{\left(n\right)}$ are asymptotically $\mathcal{D}$-free.

\end{definition}

\begin{example}

\label{exRectSp}

Consider the matrices from (\ref{mimomodel}) for a fixed $N$. We set $N_0:=N$ and let $$M^{\left(N\right)}=\sum_{i=0}^{k}N_{i}^{\left(N\right)}.$$ 
We consider pairwise orthogonal projections $P_0,\dots P_k$ in the space $\left(\mathcal{A}_{M},\tau_{M}\right)$ of $M\times M$ random matrices, where each $P_i$ is a projection to a subspace of size $N_i$. Then we can embed our matrices  $U_{1},T_{1},H_{1},H_{1}^{*}\dots,U_{k},T_{k}, H_{k},H_{k}^{*}$ as simple elements in $\mathcal{A}_{M}$, as illustrated below:
$$
\begin{picture}(120,120)

\put(0,0){\line(0,1){120}}
\put(25,0){\line(0,1){120}}
\put(120,0){\line(0,1){120}}
\put(60,95){\line(0,1){25}}
\put(60,60){\line(0,1){35}}
\put(75,0){\line(0,1){45}}
\put(75,95){\line(0,1){25}}
\put(5,105){$P_0$}
\put(29,83){$T_1, P_1$}
\put(29,67){$U_1, U_1^*$}
\put(35,105){$H_1$}
\put(92,105){$H_k$}
\put(5,75){$H_1^*$}
\put(5,20){$H_k^*$}
\put(62,48){$\ddots$}
\put(10,48){$\vdots$}
\put(62,105){$\dots$}
\put(83,28){$T_k, P_k$}
\put(83,12){$U_k, U_k^*$}
\put(0,0){\line(1,0){120}}
\put(0,95){\line(1,0){120}}
\put(0,120){\line(1,0){120}}
\put(25,60){\line(1,0){35}}
\put(0,60){\line(1,0){25}}
\put(75,45){\line(1,0){45}}
\put(0,45){\line(1,0){25}}

\end{picture}
$$
We denote by $\tilde{A}$ the embedding of $A$. A simple element in a (random) matrix rectangular space will be called a \textit{simple matrix}. 

In order to achieve asymptotic freeness we need some requirements on the asymptotic distribution of the involved deterministic matrices (which we try to keep to a minimum). We will assume:
\begin{itemize}
 \item[(A4)] 
Each matrix $T_{i}$ converges in distribution w.r.t. $\frac{1}{N_{i}}\mathrm{Tr}$, and $H_{1}H_{1}^{*},\dots,H_{k}H_{k}^{*}$ converge in joint distribution w.r.t. $\frac{1}{N}\mathrm{Tr}$.
\end{itemize}
On these rectangular spaces, we are prevented from multiplying matrices that do not fit together (e.g. $\tilde{H_{i}}\tilde{H_{j}}=0$). These kind of mixed moments do not show up in the distribution of $\Phi^{\left(N\right)}$, so we will still be able to calculate this distribution by working in these rectangular spaces. The advantage here is that even when we are only asking for the existence of the limiting joint moments of $H_{1}H_{1}^{*},\dots,H_{k}H_{k}^{*}$, we will actually have that $\tilde{H_{1}},\dots,\tilde{H_{k}}$ have a limiting joint $*$-distribution.

\end{example}

We recall results by Benaych-Georges in \cite{Be2} generalizing, in the framework of rectangular spaces, Voiculescu's results on asymptotic freeness of square matrices.

\begin{proposition}

\label{thm:Unitary}

Let $k\geq1$ be fixed and for each $N\geq1$, let $\left(\mathcal{A}_{N},\tau_{N}\right)$ be a $(P_{1}^{\left(N\right)},\dots,P_{k}^{\left(N\right)})$-rectangular probability space of random matrices, such that $\tau_{N}(P_{i}^{\left(N\right)})\to c_{i}\in(0,1]$, $i=1,\dots,k$. Let $(U_{i}^{\left(N\right)})_{i\geq1}$ be a collection of independent simple random matrices in $\mathcal{A}_{N}$, such that for every $N$, $U_{i}^{\left(N\right)}\in\mathcal{A}_{N}^{\left(j\left(i\right)\right)}$ is a Haar-distributed unitary ensemble in the compressed space $(\mathcal{A}_{N}^{\left(j\left(i\right)\right)},\tau_{N}^{\left(j\left(i\right)\right)})$ for some $1\leq j\left(i\right)\leq k$.

Furthermore, let $\mathcal{D}_{N}=(D_{i}^{\left(N\right)})_{i\geq1}$ be a collection of simple deterministic matrices, such that for all $N$,  $D_{i}^{\left(N\right)}\in\mathcal{A}_{N}^{\left(r\left(i\right),s\left(i\right)\right)}$ for some $1\leq r\left(i\right),s\left(i\right),\leq k$, and assume that $(D_{i}^{\left(N\right)})_{i\geq1}$ converge in $\mathcal{D}$-distribution.

Then $\mathcal{D}_{N},U_{1}^{\left(N\right)},U_{2}^{\left(N\right)},\dots$ are asymptotically $\mathcal{D}$-free.
\end{proposition}

\begin{remark}

As should be expected, if in the preceding theorem one considers Gaussian or non self-adjoint Gaussian simple matrices instead of Haar unitary simple matrices, the same $\mathcal{D}$-freeness results hold, where the corresponding random matrices converge to semicircular or circular simple elements in a limiting rectangular space.

In \cite{Be1} the case where $\tau_{N}\left(P_{1}\right)\to0$ is also addressed.

From the combinatorial point of view, the computation of the asymptotic scalar-valued moments of a product of rectangular matrices is performed in the same way as in the usual, square case, namely, it can be obtained as a sum of products of moments of smaller degree, running over certain collections of non-crossing partitions. These phenomena will be better observed later in our proof of the asymptotic freeness of randomly rotated deterministic matrices.

Roughly speaking, the only difference from the square case is that one needs to scale each moment by taking the different sizes of the matrices into consideration. This scaling is automatically performed by the use of the $\mathcal{D}$-distribution, which puts the right weights on each moment. In particular, the order of convergence to the moments of free operators is the same as in the square case.

We also want to remark that, by following the arguments for proving asymptotic freeness between deterministic and Wigner matrices (see \cite{AnGuZe, Mingo-Speicher}), the previous results should admit a straight-forward generalization to the case where the Gaussian matrices are replaced by Wigner matrices. 

\end{remark}

\subsection{Randomly Rotated Matrices.}

A well-known result from scalar-valued free probability states that one can consider a pair of collections  $\mathcal{D}_{1}^{\left(N\right)},\mathcal{D}_{2}^{\left(N\right)}$ of deterministic matrix ensembles, where each collection of ensembles converges in distribution (with respect to $\frac{1}{N}\mathrm{Tr}$). Then if we randomly rotate all the matrices in one of the collections by conjugating with a Haar-distributed unitary matrix, the resulting collections $\mathcal{D}_{1}^{\left(N\right)},U_{N}\mathcal{D}_{2}^{\left(N\right)}U_{N}^{*}$ are asymptotically free.
 
We do not require the existence of an asymptotic joint distribution of $\mathcal{D}_{1}^{\left(N\right)},\mathcal{D}_{2}^{\left(N\right)}$. It follows from the Weingarten formulas for the joint distribution of the entries of Haar-distributed matrices, that these mixed moments do not appear in the calculations, see, e.g., \cite{Co,NiSp}.

The following theorem is a generalization of the previous result in the rectangular framework.

\begin{theorem}\label{thm:main-rectangular}

Let $k\geq1$ be fixed and for each $N\in\mathbb{N}$, let $\left(\mathcal{A}_{N},\mathbb{E}\circ\frac{1}{N}\mathrm{Tr}\right)$ be a space of $N\times N$ random matrices with the structure of a $(P_{1}^{\left(N\right)},\dots,P_{k}^{\left(N\right)})$-rectangular probability space, with simple elements  $U_{i}^{\left(N\right)}\in\mathcal{A}_{N}^{\left(i\right)}$ such that:
\begin{itemize}
 \item[i)] 
$\lim_{N\to\infty}\frac{1}{N}\mathrm{Tr}(P_{i}^{(N)})=c_{i}>0$,
\item[ii)] 
$U_{1}^{\left(N\right)},\dots,U_{k}^{\left(N\right)}$ are independent (entrywise). Each $U_{i}$ is a Haar-distributed unitary random matrix in the compressed space  $(\mathcal{A}_{N}^{\left(i\right)},\tau_{N}^{\left(i\right)})$. 
\end{itemize}
Put $\mathbf{U}_{N}:=U_{1}^{\left(N\right)}+\dots+U_{k}^{\left(N\right)}$ and let $$\mathcal{D}_{1}^{\left(N\right)}=\left\{ C_{1}^{\left(N\right)},\dots,C_{p}^{\left(N\right)}\right\}\qquad\text{and}\qquad \mathcal{D}_{2}^{\left(N\right)}=\left\{ D_{1}^{\left(N\right)},\dots,D_{q}^{\left(N\right)}\right\}$$ be collections of simple deterministic matrices, each one with asymptotic $\mathcal{D}$-distribution. Consider $$ \mathbf{U}_{N}\mathcal{D}_{2}^{\left(N\right)}\mathbf{U}_{N}^{*}:=\left(\mathbf{U}_{N}D\mathbf{U}_{N}^{*}:D\in\mathcal{D}_{2}^{\left(N\right)}\right).$$ Then $\mathcal{D}_{1}^{\left(N\right)}$ and $\mathbf{U}_{N}\mathcal{D}_{2}^{\left(N\right)}\mathbf{U}_{N}^{*}$ are asymptotically $\mathcal{D}$-free.

\end{theorem}

If we assume that $\mathcal{D}_{1}^{\left(N\right)},\mathcal{D}_{2}^{\left(N\right)}$ have a joint distribution, then the previous result is an easy corollary from Proposition \ref{thm:Unitary}. Without such an assumption, the proof of Theorem \ref{thm:main-rectangular} follows the lines of the square case, and we will leave it to the last section. First we want to show how it enables us to compute the asymptotic distribution of $\Phi_{N}$ from Equation \eqref{mimomodel}.

\begin{corollary}

\label{maintheo}

Let $U_{1}^{\left(N\right)},\dots,U_{k}^{\left(N\right)}$, $H_{1}^{\left(N\right)},\dots,H_{k}^{\left(N\right)}$ and $T_{1}^{\left(N\right)},\dots,T_{k}^{\left(N\right)}$ be deterministic ensembles satisfying conditions (A1)-(A4) and consider their embedings in the rectangular spaces described in Example \ref{exRectSp}. Then 
\begin{equation*}
\left(\tilde{H}_{1}^{\left(N\right)},\dots,\tilde{H}_{k}^{\left(N\right)},
\tilde{U}_{1}^{\left(N\right)}\tilde{T}_{1}^{\left(N\right)}
\tilde{U}_{1}^{\left(N\right)*},
\dots,\tilde{U}_{k}^{\left(N\right)}\tilde{T}_{k}^{\left(N\right)}
\tilde{U}_{k}^{\left(N\right)*}\right)
\end{equation*}
converges in joint $*$-distribution to a collection of elements $\left(h_{1},\dots,h_{k},t_{1},\dots,t_{k}\right)$ in a $W^{*}$-probability space $\left(\mathcal{A},\tau\right)$, which has the structure of a $\left(p_{0},\dots,p_{k}\right)$-rectangular probability space, satisfying
\begin{itemize}
 \item[i)] $\tau\left(p_{i}\right)=c_{i}$.
\item[ii)] $h_{i}\in\mathcal{A}^{\left(0,i\right)}$ and $t_{i}\in\mathcal{A}^{\left(i\right)}$, $1\leq i\leq k$.
\item[iii)] $\left\langle h_{1},\dots,h_{k}\right\rangle, t_{1},\dots,t_{k}$ are free with amalgamation over $\left\langle p_{0},\dots,p_{k}\right\rangle $.
\end{itemize}
\end{corollary}

We could use the previous corollary to describe the asymptotic distribution of $\Phi_{N}$. Instead, we will introduce in the next section an approach that allows us to work with finite $N$, motivated by our knowledge of the limiting behavior of the involved matrices.

\section{Free Deterministic Equivalents}

To simplify our exposition, we first introduce the free deterministic equivalents for the case where all the matrices are square and have the same size. Then it will be natural to generalize it to the case of rectangular matrices with different sizes.

\subsection{Square matrices}

Voiculescu's asymptotic freeness results on square random matrices state that if we consider tuples of independent random matrix ensembles, such as Gaussian, Wigner or Haar unitaries, their collective behavior in the large $N$ limit is that of a corresponding collection of free (semi-)circular and Haar unitary operators. Moreover, if we consider these random ensembles along with deterministic ensembles, having a given asymptotic distribution (with respect to the normalized trace), the corresponding limiting operators become free from the random elements.

We consider now a collection of independent random and deterministic $N\times N$ matrices: 
\begin{eqnarray*}
\mathbf{X}_{N}=\left\{ X_{1}^{\left(N\right)},\dots,X_{i\left(1\right)}^{\left(N\right)}\right\}  & : & \mbox{independent self-adjoint Gaussian matrices}\\
\mathbf{Y}_{N}=\left\{ Y_{1}^{\left(N\right)},\dots,Y_{i\left(2\right)}^{\left(N\right)}\right\}  & : & \mbox{independent non self-adjoint Gaussian matrices}\\
\mathbf{U}_{N}=\left\{ U_{1}^{\left(N\right)},\dots,U_{i\left(3\right)}^{\left(N\right)}\right\}  & : & \mbox{independent Haar-distributed unitary matrices}\\
\mathbf{D}_{N}=\left\{ D_{1}^{\left(N\right)},\dots,D_{i\left(4\right)}^{\left(N\right)}\right\}  & : & \mbox{deterministic matrices,}
\end{eqnarray*}
and a polynomial in non-commuting variables (and their adjoints) evaluated in our matrices 
\begin{equation*}
P\left(X_{1}^{\left(N\right)},\dots,X_{i\left(1\right)}^{\left(N\right)},Y_{1}^{\left(N\right)},\dots,Y_{i\left(2\right)}^{\left(N\right)},U_{1}^{\left(N\right)},\dots,U_{i\left(3\right)}^{\left(N\right)},D_{1}^{\left(N\right)},\dots,D_{i\left(4\right)}^{\left(N\right)}\right)=:P_{N}.
\end{equation*}

Relying on the above mentioned asymptotic freeness results, we can then compute the asymptotic distribution of $P_{N}$ with respect to the expected normalized trace $\mathbb{E}\circ\frac{1}{N}\mathrm{Tr}=:\tau_{N}$, (which, for the case of self-adjoint polynomials, $P_{N}^*=P_{N}$, coincides with the expected spectral distribution of $P_{N}$) by going over the limit. We know that we can find collections $\mathbf{S},\mathbf{C},\mathbf{U},\mathbf{D}$ of operators in a non-commutative probability space $\left(\mathcal{A},\tau\right)$,
\begin{eqnarray*}
\mathbf{S}=\left\{ s_{1},\dots,s_{i\left(1\right)}\right\}  & : & \mbox{free semi-circular elements}\\
\mathbf{C}=\left\{ c_{1},\dots,c_{i\left(2\right)}\right\}  & : & \mbox{free circular elements}\\
\mathbf{U}=\left\{ u_{1},\dots,u_{i\left(3\right)}\right\}  & : & \mbox{free Haar unitaries}\\
\mathbf{D}=\left\{ d_{1},\dots,d_{i\left(4\right)}\right\}  & : & \mbox{abstract elements,}
\end{eqnarray*}
such that $\mathbf{S},\mathbf{C},\mathbf{U},\mathbf{D}$ are $*$-free and the joint distribution of $d_{1},\dots d_{i\left(4\right)}$ is given by the asymptotic joint distribution of $D_{1}^{\left(N\right)},\dots,D_{i\left(4\right)}^{\left(N\right)}$. Then, the expected asymptotic distribution of $P_{N}$ is that of 
\begin{equation*}
P\left(s_{1},\dots,s_{i\left(1\right)},c_{1},\dots,c_{i\left(2\right)},u_{1},\dots,u_{i\left(3\right)},d_{1},\dots,d_{i\left(4\right)}\right)=:P_{\infty}, 
\end{equation*}
in the sense that, for all $k$,
\begin{equation*}
\lim_{N\to\infty}\tau_{N}\left(P_{N}^{k}\right)=\tau\left(P_{\infty}^{k}\right). 
\end{equation*}

In this way we can reduce the problem of the asymptotic distribution of $P_N$ to the study of the distribution of $P_{\infty}$, and this can be done by using the tools of free probability, due to the nice freeness relations exhibited by $\mathbf{S},\mathbf{C},\mathbf{U},\mathbf{D}$.

A common obstacle of this procedure is that our deterministic matrices may not have an asymptotic joint distribution. But now that we understand the limiting behavior of these random and deterministic matrices, it is natural to consider, for a fixed $N$, the corresponding ``free model'' 
\begin{equation*}
P\left(s_{1},\dots,s_{i\left(1\right)},c_{1},\dots,c_{i\left(2\right)},u_{1},\dots,u_{i\left(3\right)},d_{1}^{\left(N\right)},\dots,d_{i\left(4\right)}^{\left(N\right)}\right)=:P_{N}^{\square} 
\end{equation*}
where, just as before, the random matrices are replaced by the corresponding free operators in some space $\left(\mathcal{A}_{N},\varphi_{N}\right)$, but now we let the distribution of $d_{1}^{\left(N\right)},\dots,d_{i\left(4\right)}^{\left(N\right)}$ be exactly the same as the one of $D_{1}^{\left(N\right)},\dots,D_{i\left(4\right)}^{\left(N\right)}$ w.r.t $\frac{1}{N}\mathrm{Tr}$.
The free model $P_{N}^{\square}$ will be called the \textit{free deterministic equivalent} for $P_{N}$.

Once again we will be able to deal with $P_{N}^{\square}$ by the means of free probability. The difference between the distribution of $P_{N}^{\square}$ and the expected distribution of $P_{N}$ is given by the deviation from freeness of $\mathbf{X}_{N},\mathbf{Y}_{N},\mathbf{U}_{N},\mathbf{D}_{N}$, the deviation of  $\mathbf{X}_{N},\mathbf{Y}_{N}$ from being free (semi)-circular systems, and the deviation of $\mathbf{U}_{N}$ from a free system of Haar unitaries (for the case of $\mathbf{U}_{N}$, the matrices are already Haar unitaries individually, but they are not free). Of course for large $N$ these deviations get smaller and thus the distribution of $P_{N}^{\square}$ becomes a better approximation for the expected distribution of $P_{N}$. (By the concentration of measure phenomena, the difference between the empirical distribution of $P_{N}$ and its expectation will
typically also tend to zero in the large $N$ limit.)

Let us denote by $G_N$ the Cauchy transform of $P_N$ and by $G^\square_N$ the Cauchy transform of the free deterministic equivalent $P_N^\square$. Then the usual asymptotic freeness estimates
show that moments of $P_N$ are, for large $N$, close to corresponding moments of $P_N^\square$
(where the estimates involve also the operator norms of the deterministic matrices). This means that for $N\to\infty$ the difference between the Cauchy transforms $G_N$ and $G_N^\square$ goes to zero, even if there does not exist individual limits for both Cauchy transforms.
(In the appendix we will show how to extend estimates for the difference of Cauchy transforms in a neighborhood of infinity -- corresponding to knowledge about the moments of the involved distributions -- to a region close to the real axis; which can then be used to compare the respective distribution functions, or to estimate the Kolmogorov distance between the distributions.)

In the engineering literature the notion of a deterministic equivalent (apparently going back to Girko \cite{Gir}, see also \cite{HLN}) has recently gained quite some
interest. This deterministic equivalent consists in replacing the Cauchy transform $G_N$ of the considered random matrix model (for which no analytic solution exists) by a function $\hat G_N$ which is defined as the solution of a specified system of equations. The specific form of those equations is determined in an ad hoc way, depending on the considered problem, by making approximations for the equations of $G_N$, such that one gets a closed system of equations.  

In all examples of deterministic equivalents we know of it turns out that actually the Cauchy transform of our free deterministic equivalent is the solution to the equations of the deterministic equivalents, i.e., that $\hat G_N=G_N^\square$. We think that our definition of a deterministic equivalent gives a more conceptual approach and shows clearly how this relates with free probability theory. In some sense this indicates that the only meaningful way to get a closed system of equations when dealing with random matrices is to replace the random matrices by free variables. We want to point out that the same phenomena was essentially also observed in 
\cite{Neu-Sp} in the context of the so-called CPA approximation (a kind of mean-field approximation) for the Anderson model from statistical physics. In our present language their result can be rephrased as saying that the free deterministic equivalent of the Anderson model is given by the CPA approximation.

Since in engineering applications rectangular matrices are quite typical, it is important to be able to have the notion of a free deterministic equivalent also for problems involving rectangular matrices. In the next section we will show how the above can be generalized to the rectangular setting, thus hopefully refuting the common perception that rectangular matrices cannot be dealt with in free probability.

\subsection{Rectangular matrices}

In order to include models with rectangular matrices of different sizes, we have to consider slightly more complicated situations which take place in rectangular probability spaces. The main idea is that the space is split into pieces $(\mathcal{A}_{N}^{\left(i,j\right)})_{1\leq i,j\leq k}$ and one can think of a rectangular matrix as an element in one of these pieces; in particular, square matrices will be regarded as elements in the compressed spaces $\mathcal{A}_{N}^{\left(i\right)}$.
 
We can still use all the results on asymptotic freeness for square Gaussian, Haar unitary and deterministic matrices in this framework, by just working in the corresponding compressed space $(\mathcal{A}_{N}^{\left(i\right)},\tau_{N}^{\left(i\right)})$, with the novelty that different sized square random matrices are  $\mathcal{D}$-free from each other and from deterministic rectangular matrices. 

To obtain the free deterministic equivalent, we will replace the random square matrices by free operators ((semi-)circulars and Haar unitaries) in the corresponding compressed spaces, just in the way described in the square case. When it comes to rectangular deterministic matrices in $\mathcal{A}^{\left(i,j\right)}$, we will replace them by operators, keeping the same distribution, but prescribing them to be $\mathcal{D}$-free from the (semi-) circulars and the Haar unitaries.

It is important to impose some restrictions in our non-commutative polynomials for this rectangular situation. We should only treat polynomials where the sums and the products are performed in such a way that the sizes fit. The reason for this is that our rectangular spaces prescribe non-compatible products of matrices to be zero, and hence, such products should not appear in our polynomials. 
\begin{remark}

In order to simplify notation, when the considered non-commutative polynomial $P$ depends on a randomly rotated deterministic matrix  $U^{\left(N\right)}D^{\left(N\right)}U^{\left(N\right)*}$ and not actually on $U^{\left(N\right)},D^{\left(N\right)}\in\mathcal{A}^{\left(i\right)}$ individually, we will just write $d^{\left(N\right)}$ (instead of $ud^{\left(N\right)}u^{*}$) in our free deterministic equivalent, keeping in mind that $d^{\left(N\right)}$ will be free from all other deterministic matrices.

\end{remark}

\section{An Example}

From our discussion on the previous section, the free deterministic equivalent for $\Phi_{N}$ from \eqref{mimomodel} will be 
\begin{equation*}
\Phi:=\Phi_{N}^{\square}=\sum_{i=1}^{k}h_{i}^{\left(N\right)}t_{i}^{\left(N\right)}
h_{i}^{\left(N\right)*}, 
\end{equation*}
where (after supressing the $N$ super-index) each $h_{i}$ and $t_{i}$ is a simple element in a $(p_{0},\dots,p_{k})$-probability space $\left(\mathcal{A},\tau\right)$, with conditional expectation $\mathbf{F}:\mathcal{A}\to\mathcal{D}$, satisfying
\begin{itemize}
 \item[i)] $\tau(p_{i})=\frac{N_{i}}{M}$ for $1\leq i\leq k$
\item[ii)] 
$h_{i}\in\mathcal{A}^{\left(0,i\right)}$, $t_{i}\in\mathcal{A}^{\left(i\right)}$, with 
\begin{equation*}
\tau^{\left(i\right)}\left(t_{i}^{m}\right)=\frac{1}{N_{i}}\mathrm{Tr}\left(T_{i}^{m}\right), 
\end{equation*}
for $1\leq i\leq k$, and for all $m\in\NN$ and all $1\leq i_1,\dots,i_m\leq k$ we have
\begin{eqnarray*}
\tau^{\left(0\right)}\left(h_{i_{1}}h_{i_{1}}^{*}\cdots h_{i_{m}}h_{i_{m}}^{*}\right) & = & \frac{1}{N}\mathrm{Tr}\left(H_{i_{1}}H_{i_{1}}^{*}\cdots H_{i_{m}}H_{i_{m}}^{*}\right)\\
\tau^{\left(i_{1}\right)}\left(h_{i_{1}}^{*}\cdots h_{i_{m}}h_{i_{m}}^{*}h_{i_{1}}\right) & = & \frac{1}{N_{i_{1}}}\mathrm{Tr}\left(H_{i_{1}}^{*}\cdots H_{i_{m}}H_{i_{m}}^{*}H_{i_{1}}\right)
\end{eqnarray*}
\item[iii)] 
$\left\{ h_{1},\dots,h_{k}\right\} ,t_{1},\dots,t_{k}$ are $\mathcal{D}$-free.
\end{itemize}
We consider the subalgebras $$\mathcal{C}:=\left\langle \mathcal{D},h_{1}h_{1}^{*},\dots,h_{k}h_{k}^{*}\right\rangle \subset\mathcal{A}\qquad \text{and}\qquad \mathcal{C}^{\left(0\right)}:=p_{0}\mathcal{C}p_{0}\subset\mathcal{A}^{\left(0\right)},$$ 
with the conditional expectations $$\mathbf{E}:\mathcal{A}\to\mathcal{C}\qquad\text{and}\qquad \mathbf{E}^{\left(0\right)}:\mathcal{A}^{\left(0\right)}\to\mathcal{C}^{\left(0\right)}$$ compatible with $\tau$ and $\tau^{\left(0\right)}$, respectively.

We are ultimately interested in the distribution of $\Phi=\Phi_{N}^{\square}$ considered as an element of $\mathcal{A}^{\left(0\right)}$. It is easy to see that $\mathbf{E}|_{\mathcal{A}^{\left(0\right)}}$ is $\mathcal{C}^{\left(0\right)}$-valued and defines a conditional expectation compatible with $\tau^{\left(0\right)}$.
Hence, by uniqueness, $\mathbf{E}^{\left(0\right)}$ must be $\mathbf{E}|_{\mathcal{A}^{\left(0\right)}}$. Then we observe that for invertible $c\in\cC$
\begin{eqnarray*}
p_{0}G_{\Phi}^{\mathcal{C}}\left(c\right)p_{0} & = & p_{0}\sum_{n\geq0}\mathbf{E}\left(c^{-1}\left(\Phi c^{-1}\right)^{n}\right)p_{0}\\
 & = & \sum_{n\geq0}\mathbf{E}\left(p_{0}c^{-1}p_{0}\left(\Phi p_{0}c^{-1}p_{0}\right)^{n}\right)\\
 & = & \sum_{n\geq0}\mathbf{E}^{(0)}\left(\left(p_{0}cp_{0}\right)^{-1}\left(\Phi\left(p_{0}cp_{0}\right)^{-1}\right)^{n}\right)\\
 & = & G_{\Phi}^{\mathcal{C}_{0}}\left(p_{0}cp_{0}\right),
\end{eqnarray*}
which follows from the commutativity of $\mathcal{C}$ with the elements of $\mathcal{D}$ (note that $p_0\in\cD$). Hence it suffices to describe the $\mathcal{C}$-valued distribution of $\Phi$. 

For this we first notice that $\mathcal{B}:=\left\langle \mathcal{D},h_{1},\dots,h_{k}\right\rangle $ and $\mathcal{N}:=\left\langle \mathcal{D},t_{1},\dots,t_{k}\right\rangle $ are $\mathcal{D}$-free, so we are in position of applying Proposition \ref{thmNiShSp}, and we obtain
\begin{align*}
\kappa_{n}^{\mathcal{B}}(h_{i_{1}}t_{i_{1}}h_{i_{1}}^{*}b_{1},&
h_{i_{2}}t_{i_2}h_{i_{2}}^{*}b_{2},
\dots,h_{i_{n}}t_{i_{n}}h_{i_{n}}^{*}) \\
&=
h_{i_{1}}\kappa_{n}^{\mathcal{B}}\left(t_{i_{1}}h_{i_{1}}^{*}b_{1}h_{i_{2}},
t_{i_{2}}h_{i_{2}}^{*}b_{2}h_{i_{3}},
\dots,t_{i_{n}}\right)h_{i_{n}}^{*}\\
 & = h_{i_{1}}\kappa_{n}^{\mathcal{D}}\left(t_{i_{1}}\mathbf{F}
\left(h_{i_{1}}^{*}b_{1}h_{i_{2}}\right),t_{i_{2}}\mathbf{F}
\left(h_{i_{2}}^{*}b_{2}h_{i_{2}}\right),
\dots,t_{i_{n}}\right)h_{i_{n}}^{*}.
\end{align*}
Since $t_{1},\dots,t_{k}$ are $\mathcal{D}$-free, the last cumulant vanishes unless $i_{1}=\dots=i_{n}$ (and hence $h_{1}t_{1}h_{1}^{*},\dots,h_{k}t_{k}h_{k}^{*}$ are  $\mathcal{B}$-free). Moreover,  in the latter case 
we have that
\begin{eqnarray*}
\kappa_{n}^{\mathcal{B}}\left(h_{i}t_{i}h_{i}^{*}b,\dots,h_{i}t_{i}h_{i}^{*}\right) & = & h_{i}\kappa_{n}^{\mathcal{D}}\left(t_{i}\mathbf{F}\left(h_{i}^{*}bh_{i}\right),\dots,t_{i}\right)h_{i}^{*}
\end{eqnarray*}
is $\mathcal{C}$-valued for all $b_{1},\dots,b_{k-1}\in\mathcal{B}$ (in particular for all $b_{1},\dots,b_{k-1}\in\mathcal{C}$), and thus, by Proposition \ref{thm:Restrictions}, the $\mathcal{C}$-cumulants of $h_{1}t_{1}h_{1}^{*},\dots,h_{k}t_{k}h_{k}^{*}$ are the restrictions of the $\mathcal{B}$-cumulants.
Thus the vanishing of the mixed $\cB$-cumulants implies also the vanishing of the mixed $\cC$-cumulants, hence $h_{1}t_{1}h_{1}^{*},\dots,h_{k}t_{k}h_{k}^{*}$ are also $\cC$-free. Furthermore, we get, for all $c\in\mathcal{C}$, and all $i=1,\dots,k$ that
\begin{eqnarray*}
\kappa_{n}^{\mathcal{C}}\left(h_{i}t_{i}h_{i}^{*}c,\dots,h_{i}t_{i}h_{i}^{*}\right) & = & \kappa_{n}^{\mathcal{B}}\left(h_{i}t_{i}h_{i}^{*}c,\dots,h_{i}t_{i}h_{i}^{*}\right)\\
 & = & h_{i}\kappa_{n}^{\mathcal{D}}\left(t_{i}\mathbf{F}\left(h_{i}^{*}ch_{i}\right),\dots,t_{i}\right)h_{i}^{*}\\
 & = & h_{i}\kappa_{n}^{\mathcal{D}}\left(t_{i}p_{i}\tau^{\left(i\right)}(h_{i}^{*}ch_{i}),\dots,t_{i}\right)h_{i}^{*}\\
 & = & h_{i}h_{i}^{*}\left(\tau^{\left(i\right)}(h_{i}^{*}ch_{i})\right)^{n-1}\kappa_{n}^{\left(i\right)}\left(t_{i},\dots,t_{i}\right).
\end{eqnarray*}
Hence
\begin{eqnarray*}
\mathcal{R}_{h_{i}t_{i}h_{i}^{*}}^{\mathcal{C}}\left(c\right) & = & \sum_{n=1}^{\infty}\kappa_{n}^{\mathcal{C}}\left(h_{i}t_{i}h_{i}^{*}c,\dots,h_{i}t_{i}h_{i}^{*}\right)\\
 & = & \sum_{n=1}^{\infty}h_{i}h_{i}^{*}\left(\tau^{\left(i\right)}\left(h_{i}^{*}ch_{i}\right)\right)^{n-1}\kappa_{n}^{\left(i\right)}\left(t_{i},\dots,t_{i}\right)\\
 & = & h_{i}h_{i}^{*}\mathcal{R}_{t_{i}}^{\left(i\right)}\left(\tau^{\left(i\right)}\left(h_{i}^{*}ch_{i}\right)\right)
\end{eqnarray*}
Since $h_{1}t_{1}h_{1}^{*},\dots,h_{k}t_{k}h_{k}^{*}$ are $\cC$-free we get for the $\cC$-valued $\mathcal{R}$-transform of $\Phi$ that
$$\mathcal{R}_{\Phi}^\cC(c)=\sum_{i=1}^{k}\mathcal{R}_{h_{i}t_{i}h_{i}^{*}}^{\mathcal{C}}(c).$$
Now
we apply all this to our functional relation between the $\mathcal{C}$-valued Cauchy and $\mathcal{R}$-transform. We obtain
\begin{eqnarray*}
G_{\Phi}^{\mathcal{C}}\left(c\right) & = & \left(\mathcal{R}_{\Phi}^{\mathcal{C}}\left(G_{\Phi}^{\mathcal{C}}\left(c\right)\right)-c\right)^{-1}\\
 & = & \left(\sum_{i=1}^{k}\mathcal{R}_{h_{i}t_{i}h_{i}^{*}}^{\mathcal{C}}\left(G_{\Phi}^{\mathcal{C}}\left(c\right)\right)-c\right)^{-1}\\
 & = & \left(\sum_{i=1}^{k}h_{i}h_{i}^{*}\mathcal{R}_{t_{i}}^{\left(i\right)}\left(\tau^{\left(i\right)}\left(h_{i}^{*}G_{\Phi}^{\mathcal{C}}\left(c\right)h_{i}\right)\right)-c\right)^{-1}
\end{eqnarray*}
and therefore 
\begin{eqnarray} \label{eq:fixpointforG} \notag
G_{\Phi}^{\mathcal{C}^{\left(0\right)}}(p_{0}cp_{0}) & = &  p_{0}G_{\Phi}^{\mathcal{C}}\left(c\right)p_{0}\\ \notag
 & = & p_{0}\left(\left(\sum_{i=1}^{k}h_{i}h_{i}^{*}\mathcal{R}_{t_{i}}^{\left(i\right)}\left(\tau^{\left(i\right)}\left(h_{i}^{*}G_{\Phi}^{\mathcal{C}}\left(c\right)h_{i}\right)\right)-c\right)^{-1}\right)p_{0}\\ \notag
 & = & \left(p_{0}\sum_{i=1}^{k}h_{i}h_{i}^{*}\mathcal{R}_{t_{i}}^{\left(i\right)}\left(\tau^{\left(i\right)}\left(h_{i}^{*}p_{0}G_{\Phi}^{\mathcal{C}}\left(c\right)p_{0}h_{i}\right)\right)p_{0}-p_{0}cp_{0}\right)^{-1}\\ 
 & = & \left(\sum_{i=1}^{k}h_{i}h_{i}^{*}\mathcal{R}_{t_{i}}^{\left(i\right)}\left(\frac{N}{N_{i}}\tau^{\left(0\right)}\left(G_{\Phi}^{\mathcal{C}^{\left(0\right)}}
\left(p_{0}cp_{0}\right)h_ih_{i}^{*}\right)\right)-p_{0}cp_{0}\right)^{-1},
\end{eqnarray}
where in the second step we introduced the $p_0$ in the argument of $\tau^{(i)}$ by using $p_0h_i=h_i$; in the third step we expressed $\tau^{(i)}$ in terms of $\tau^{(0)}$ according to their definition.

Consider now $c\in\cC^{(0)}$, so that $p_0cp_0=c$. Then, 
with the definition
$$f_j(c):=\frac{N}{N_{j}}\tau^{(0)}\left(G_{\Phi}^{\mathcal{C}_{0}}
\left(c\right)h_{j}h_{j}^{*}\right)$$
this becomes
$$
G_{\Phi}^{\mathcal{C}^{\left(0\right)}}(c) = \left(\sum_{i=1}^{k}h_{i}h_{i}^{*}\mathcal{R}_{t_{i}}^{\left(i\right)}\left(
f_i(c)
\right)-c\right)^{-1},
$$
and thus we have for the $f_j(c)$ the system of equations
$$f_j(c)=
\frac{N}{N_{j}}\tau^{(0)}\left[
\left(\sum_{i=1}^{k}h_{i}h_{i}^{*}\mathcal{R}_{t_{i}}^{\left(i\right)}\left(
f_i(c)
\right)-c\right)^{-1}
h_{j}h_{j}^{*}\right]
$$

In order to get now the scalar-valued Cauchy transform $G_\Phi(z)$ (for $z\in\CC$) of
our free deterministic equivalent $\Phi=\Phi_N^\square$ we project the
operator-valued Cauchy transform to the scalar level:
$$G_{\Phi}(z)=\tau^{\left(0\right)}\left(G_{\Phi}^{\mathcal{C}^{(0)}}
\left(z\mathbf{1}_{\mathcal{C}^{\left(0\right)}}\right)\right)=\tau^{(0)}\left[
\left(\sum_{i=1}^{k}h_{i}h_{i}^{*}\mathcal{R}_{t_{i}}^{\left(i\right)}\left(
f_i(z)
\right)-z \mathbf{1}_{\mathcal{C}^{\left(0\right)}}\right)^{-1}\right],
$$
where we also put $f_i(z):=f_i(z\mathbf{1}_{\mathcal{C}^{\left(0\right)}})$.

If we rewrite all this in terms of the matrices $H_i$ and $T_i$ we arrive finally at
\begin{eqnarray*}
G_{\Phi}\left(z\right) & = & \frac{1}{N}\mathrm{Tr}\left[\left(\sum_{i=1}^{k}H_{i}H_{i}^{*}\mathcal{R}_{T_{i}}\left(f_{i}\left(z\right)\right)-zI_{N}\right)^{-1}\right]
\end{eqnarray*}
where the $f_j(z)$ are determined by the system of equations ($j=1,\dots,k$)
\begin{equation*}
f_{j}\left(z\right)=\frac{1}{N_{j}}\mathrm{Tr}\left[\left(\sum_{i=1}^{k}H_{i}H_{i}^{*}
\mathcal{R}_{T_{i}}\left(f_{i}\left(z\right)\right)-zI_{N}\right)^{-1}H_{j}H_{j}^{*}\right]
\end{equation*}

These equations are equivalent to the ones showing up in \cite{CouDeHo}.
(Since there they do not use the $\cR$-transform of the matrices $T_i$, this information has to be encoded in another set of equations in their approach.)

One should of course also consider the question whether those equations determine the $f_j(z)$ uniquely, within a suitably chosen class of functions.
This questions is answered affirmatively for the present example in \cite{CouDeHo}. Let us also point out that it is reasonable to expect that such a question (in the form whether the operator-valued equation \eqref{eq:fixpointforG} has a unique fixed point in the class of Cauchy transforms) can be treated in a suitably general frame by extending the approach from
\cite{HRS}.
\begin{remark}

In \cite{CouDeHo}, they consider the seemingly more general problem, with deterministic $n_{i}\times n_{i}$ matrices $T_{i}$ and random matrices $W_{i}=U_{i}P_{i}$ which are rectangular $N_{i}\times n_{i}$, obtained by considering only the first $n_{i}$ columns of a square $N_{i}\times N_{i}$ Haar-distributed unitary matrix.

If we let $\hat{T}_{i}$ be the inclusion of $T$ in the space of $N_{i}\times N_{i}$ by just filling up with zeros, we have
\begin{eqnarray*}
U_{i}\hat{T}_{i}U_{i}^{*} & = & U_{i}P_{i}\hat{T}_{i}P_{i}U_{i}^{*}\\
 & = & W_{i}T_{i}W_{i}^{*}.
\end{eqnarray*}
The convergence in distribution of $T$ w.r.t. $\frac{1}{n_{i}}\mathrm{Tr}$ is equivalent to the convergence in distribution of $\hat{T}_{i}$ w.r.t. $\frac{1}{N_{i}}\mathrm{Tr}$ (the moments just get scaled by a constant). Hence it will be enough to solve the case when $N_{i}=n_{i}$.

\end{remark}

\section{Asymptotic $\mathcal{D}$-freeness of Randomly Rotated Matrices}
In this section we will prove Theorem \ref{thm:main-rectangular}.
For notation, definitions and a detailed exposition of the results used in this section, we refer to the last lecture of \cite{NiSp}, where the square case is treated. We will only point out those aspects of the proof which differ from the square case. First we will compute the joint $\mathcal{D}$-moments of $\mathcal{D}_1,\mathbf{U}\mathcal{D}_2\mathbf{U}^*$ assuming that the desired asymptotic $\mathcal{D}$-freeness holds. Then we compute our asymptotic joint moments by the use of the formulas relating the joint distributions of the entries of Haar unitary matrices and the Weingarten function. We realize that both computations coincide.

\subsection{Computation of Operator-valued moments.}

Suppose that $\left(\mathcal{A},\tau\right)$ is a $(p_1\dots,p_k)$-rectangular probability space and that $c_1,\dots,c_p$, $d_1,\dots,d_q$ are simple elements such that $\mathcal{D}_1:=\langle c_1,\dots,c_q\rangle$, $\mathcal{D}_2:=\langle d_1,\dots,d_p\rangle$ are $\mathcal{D}$-free. The goal of this section is to express the mixed $\mathcal{D}$-moments of $\mathcal{D}_1$, $\mathcal{D}_2$ in terms of the individual $\mathcal{D}$-moments of $\mathcal{D}_1$ and $\mathcal{D}_2$

By linearity we may only consider monomials. The first observation is that a product of simple elements is simple as well. So a general $\mathcal{D}$-moment will be of the form
\begin{equation*}
\mathbf{F}(D_1C_1\dots D_nC_n),
\end{equation*}
where $D_j\in \mathcal{A}^{(r(j),r'(j))}$ is a product of the $d_i$'s and $C_i$ is a product of the $c_i$'s. In order to have a non-zero moment, we must have that $C_j\in \mathcal{A}^{(r'(j),r(j+1))}$. In the following we only consider such products. We denote by $V_m:=\{x:r(x)=m\}\subseteq[n]$.

If we consider a totally ordered set $S=S_1\cup S_2$, where $S_1\cap S_2=\emptyset$ and we let $\{V_1,\dots ,V_s\}=\pi_1$ be a partition of $S_1$ and $\{W_1,\dots ,W_t\}=\pi_2$ be a partition of $S_2$, we will write $\pi_1 \cup \pi_2$ for the partition $\{V_1,\dots,V_s,W_1\dots,W_t\}$ of $S$. For a partition $\pi \in NC(n)$ we consider its \textit{Kreweras
complement}, defined as the largest partition $K(\pi) \in NC(\overline{1},\dots ,\overline{n})$ such that $\pi \cup K(\pi) \in
NC(1,\overline{1},\dots ,n,\overline{n})$. For a pair of non-crossing
partitions $\sigma \leq \pi \in NC(n)$, we have that $\sigma \cup K(\pi)
\in NC(1,\overline{1},\dots ,n,\overline{n})$. We will be interested in
pairs
$\sigma \leq \pi$ such that $$ r'(i)=r(\sigma (i)), r(i)=r'(K(\pi )(i-1)),
\qquad  1\leq i\leq n ,$$ where the arithmetic is performed modulo $n$ and
a partition $\pi$ is viewed here as the permutation in $S_n$ with cycle
decomposition prescribed by the block structure of $\pi$ (ordered
increasingly). The collection of such pairs of partitions will be denoted
by $NC(n)_{r,r'} \subset NC(n)^2$. For such pairs we define the weight $$
\omega (\sigma, \pi)_{r,r'}:=\left(\prod_{\substack {i \in [n] \\
i<\sigma(i)}} \tau \left(
p_{r'(i)}\right)^{-1}\right)\left(\prod_{\substack {j \in [n] \\ j<K(\pi
)(j)}} \tau \left( p_{r(j+1)}\right)^{-1}\right)$$

\begin{proposition}

\label{opvalmo}

Let $\mathbf{F}(D_1C_1\dots D_nC_n)$, $r$ and $r'$ be as above. Then
\begin{eqnarray*}
& & \mathbf{F}(D_1C_1\dots D_nC_n) \\ &=& \frac{p_{r(1)}}{\tau(p_{r(1)})}\sum_{\substack{ \sigma,\pi \in NC(n)_{r,r'}  \\ \sigma \leq \pi }} \omega(\sigma , K(\pi))_{r,r'} \cdot \tau_{K(\pi)}[D_1,\dots,D_n]\cdot \tau_{\sigma}[C_1,\dots,C_n]\cdot\mu _n[\sigma ,\pi]
\end{eqnarray*}
\end{proposition}

\begin{proof}
We use the operator-valued moment-cumulant formulas
\begin{equation*}
\mathbf{F}(D_1C_1\dots D_nC_n)=\sum_{\rho \in NC(1,\overline{1},\dots,n,\overline{n})} \kappa_{\rho}^{\mathcal{D}}[D_1,C_1,\dots,D_n,C_n].
\end{equation*}
By freeness of $\mathcal{D}_1$, $\mathcal{D}_2$, we see that for a non-vanishing contribution to the sum, $\rho $ must be of the form $\pi \cup \overline{\pi}$, where $\pi \in NC(n)$ and $\overline{\pi}\leq K(\pi)$. By fixing one $\pi \in NC(n)$ we may sum over all partitions $\overline{\pi}\leq K(\pi)$, to obtain an expression where the sum runs over $\pi \in NC(n)$ and the sumand consists on evaluating in a nested way (indicated by $\pi \cup K(\pi)$) the operator-valued cumulants of the $D$'s and the operator-valued moments of the $C$'s. Then, for each fixed $\pi$ we use the cumulant-moment formula to express the cumulants in terms of moments, obtaining
\begin{equation*}
\mathbf{F}(D_1C_1\dots D_nC_n)=\sum_{\sigma \leq \pi \in NC(n)} \mathbf{F}_{\sigma \cup K(\pi)}[D_1,C_1,\dots,D_n,C_n]\cdot\mu_{NC}[\sigma, \pi].
\end{equation*}
Finally it is easy to observe that we only have a non-vanishing contribution when $(\sigma ,\pi) \in NC(n)_{r,r'}$. In these cases we can use formula (\ref{eq:condexp}) to get the last expression.
\end{proof}

\subsection{Asymptotic Joint Distribution of Haar Unitary Matrices of Different Sizes.}

We conclude the proof of our theorem by showing that the asymptotic $\mathcal{D}$-moments of $\mathcal{D}_1^{(N)},\mathbf{U}_N\mathcal{D}_2^{(N)}\mathbf{U}^*_N$ are computed as in Proposition \ref{opvalmo}.

First we fix (and omit) $N$ for simplicity of notation. One can see that a general non-vanishing operator-valued moment is of the form
$$\mathbf{F}(U_{r(1)}D^{(1)}U^*_{r'(1)}C^{(1)}\dots U_{r(n)}D^{(n)}U^*_{r'(n)}C^{(n)}),$$ where $D^{(j)}\in\mathcal{A}^{(r(j),r'(j))}$ and $C^{(j)}\in\mathcal{A}^{(r'(j),r(j+1))}$ are products of simple elements in $\mathcal{D}_1^{(N)}$ and $\mathcal{D}_2^{(N)}$, respectively. Then we have
\begin{eqnarray*}
& & \mathbf{F}(U_{r(1)}D^{(1)}U^*_{r'(1)}C^{(1)}\dots U_{r(n)}D^{(n)}U^*_{r'(n)}C^{(n)})\\
&=& \frac{P_{r(1)}}{\tau (P_{r(1)})}\frac{1}{N}\mathrm{Tr}\otimes \mathbb{E} \left(U_{r(1)}D^{(1)}U^*_{r'(1)}C^{(1)}\dots U_{r(n)}D^{(n)}U^*_{r'(n)}C^{(n)} \right) \\ 
&=& \frac{P_{r(1)}}{\tau (P_{r(1)})}\frac{1}{N} \sum_{\substack {s\in [n] \\ i_s,j_s\in [N_{r(s)}]  \\ i'_s,j'_s\in [N_{r'(s)}] }} \mathbb{E}\left(u^{(r(1))}_{i_1j_1}d^{(1)}_{j_1j'_1}\overline{u}^{(r'(1))}_{i'_1j'_1}c^{(1)}_{i'_1i_2}\dots u^{(r(n))}_{i_nj_n}d^{(n)}_{j_nj'_n}\overline{u}^{(r'(n))}_{i'_nj'_n}c^{(n)}_{i'_ni_1}\right)
\end{eqnarray*}
The only difference from \cite{NiSp} is that the expectation of the product of the entries of our unitaries will be factored by the assumption of independence between the Haar Unitaries. We can apply Lemma 2.2 from \cite{Xu} to each of these factors. In the language of the permutations considered in \cite{NiSp} this factorization of the expectation will be translated to the conditions on $\alpha, \beta \in S_n$, that  $r_i=r'_{\alpha(i)}$, $r_i=r'_{\beta(i)}$, $i=1,\dots ,n$. We denote by $S_{n,r,r'}\subset S_n$ the subset of such permutations. If we write $\alpha_m:=\alpha\mid_{V_m}$, $\beta_m:=\beta\mid_{V_m}$, we notice that $\alpha_m^{-1}\beta_m\in S_{V_m}$. Then one can show that the sum of the right hand side of the last equation is actually
\begin{eqnarray*}
& &\sum_{\substack {s\in [n] \\ i_s,j_s\in [N_{r(s)}]  \\ i'_s,j'_s\in [N_{r'(s)}] }} d^{(1)}_{j_1j'_1}\dots d^{(n)}_{j_nj'_n}c^{(1)}_{i'_1i_2}\dots c^{(n)}_{i'_ni_1} \prod_{m=1}^k\mathbb{E}\left(\prod_{\substack {s\in [n] \\ r(s)=m}}u^{(m)}_{i_sj_s} \prod_{\substack {s'\in [n] \\ r'(s')=m}}\overline{u}^{(m)}_{i'_{s'}j'_{s'}}\right) \\
&=& \sum_{\alpha, \beta \in S_{n,r,r'}}\mathrm{Tr}_{\alpha}[D^{(1)},\dots ,D^{(n)}]\cdot\mathrm{Tr}_{\beta^{-1}\gamma}[C^{(1)},\dots ,C^{(n)}]\cdot\prod_{m=1}^k\mathrm{Wg}(N_m,\alpha_m^{-1}\beta_m) \\
&=& \sum_{\alpha, \beta \in S_{n,r,r'}}N^{\#(\alpha)+\#(\beta^{-1}\gamma)}\cdot\tau_{\alpha}[D^{(1)},\dots ,D^{(n)}]\cdot\tau_{\beta^{-1}\gamma}[C^{(1)},\dots ,C^{(n)}]\cdot\prod_{m=1}^k\mathrm{Wg}(N_m,\alpha_m^{-1}\beta_m),
\end{eqnarray*}
where $\gamma\in S_n$ is the long cycle $\gamma=(1,2,\dots,n-1,n)$,  $\#(\alpha)$ denotes the number of cycles of $\alpha$ and each factor $\mathrm{Wg}$ is the Weingarten function on $S_{V_m}\cong S_{|V_m|}$. Since $\sum \#(\alpha_m^{-1} \beta_m)=\#(\alpha^{-1} \beta)$, all further remarks on the order of the contribution remain the same as in the square case. 

For a non-vanishing limit, we must have again $\alpha, \beta \in
S_{NC}(n)$, $\alpha \leq \beta$. Such pairs of permutations are known to
be in one-to-one correspondence with pairs $P_\alpha \leq P_\beta$, of non-crossing partitions. One can see that the image of $S_{n,r,r'}$
under $P$ is exactly $NC(n)_{r,r'}$. Hence the sum here can be thought as
running over the same objects as in Proposition \ref{opvalmo}. For these
pairs of permutations, we get the limiting contribution of
\begin{multline*}
N^{\#(\alpha)+\#(\beta ^{-1}\gamma)-1}\prod_{m=1}^k \phi _{V_m}(\alpha_m ^{-1}\beta_m)N_m^{\#(\alpha_m^{-1}\beta_m)-2|V_m|}=\\
 N^0 \prod_{m=1}^k \phi _{V_m}(\alpha_m ^{-1}\beta_m) \prod_{m=1}^k (\tau(P_m))^{\#(\alpha_m^{-1}\beta_m)-2|V_m|},
\end{multline*}
Where $\phi$ is the coefficient of the leading term in the expansion of $\mathrm{Wg}$. It is not hard to see that the second product coincides with the weight function $\omega(P_\alpha ,P_\beta)_{r,r'}$ of the corresponding partitions.

Finally, we observe that if $\alpha, \beta \in S_{NC(n)}$, $\alpha \leq \beta$, then $\alpha^{-1} \beta \in S_{NC(n)}$. One shows that 
\begin{eqnarray*}
\prod_{m=1}^k \phi _{V_m}(\alpha_m ^{-1}\beta_m)&=&\prod_{m=1}^k \mu_{NC(V_m)}[\mathbf{0}_{V_m},P_{\alpha^{-1}_m \beta_m}]\\
&=& \mu_n \left(\prod_{m=1}^k[\mathbf{0}_{V_m},P_{\alpha^{-1}_m \beta_m}]\right)\\
&=& \mu_n[\mathbf{0}_n,P_{\alpha^{-1} \beta}]\\
&=& \mu_n[P_{\alpha},P_{\beta}]
\end{eqnarray*}
Hence, all the factors appearing in our non-vanishing terms correspond with the ones in Proposition \ref{opvalmo}. We conclude that $\mathcal{D}_1^{(N)},\mathbf{U}_N\mathcal{D}_2^{(N)}\mathbf{U}^*_N$ are asymptotically $\mathcal{D}$-free.

\appendix

\section{Extension of Estimates for the Difference of Cauchy Transforms}

\subsection{Introduction}

Roughly speaking, we want to discuss how estimates for the difference of Cauchy transforms near infinity (i.e. on a subset $\Delta^+_R := \{z\in\C^+;\ |z| > R\}$ of the upper half-plane $\C^+ := \{z\in\C;\ \Im(z)>0\}$ with $R>0$) extend to estimates near the real line. The approach presented here is motivated by a paper of V. Kargin (\cite{Kargin}). In order to understand how powerful the method is, we will reformulate his results in a more general framework. The main goal is the following theorem:

\begin{theorem}\label{Cauchy}
Let $\beta,c\in(0,1)$ and $R,T>0$ be given. Then there exist constants $m_0>0$ and $\eta_0>0$ such that for any two probability measures $\mu$ and $\nu$ on the real line $\R$ with
$$\sup_{z\in\Delta^+_R} |G_\mu(z) - G_\nu(z)| \leq e^{-m}$$
for some $m>m_0$ the inequality
$$\max_{z\in I(m)} |G_\mu(z) - G_\nu(z)| \leq \exp\big(-cm^{1-\beta}\big)$$
holds, where
$$I(m) := i\eta_0\frac{1-\exp\big(-\frac{1}{m^\beta}\big)}{1+\exp\big(-\frac{1}{m^\beta}\big)} + [-T,T].$$
\end{theorem}

This will enable us to prove the following result about the qualitative behavior of the Kolmogorov distance:

\begin{theorem}\label{Kolmogorov}
Let $\mu$ be a probability measure with compact support contained in an interval $[-A,A]$ such that the cumulative distribution function $\F_\mu$ satisfies
\begin{equation}\label{assumption}
|\F_\mu(x+t) - \F_\mu(x)| \leq \rho |t| \qquad\text{for all $x,t\in\R$}
\end{equation}
for some constant $\rho>0$. Then for all $R>0$ and $\beta\in(0,1)$ we can find $\Theta>0$ and $m_0>0$ such that for any probability measure $\nu$ with compact support contained in $[-A,A]$, which satisfies
$$\sup_{z\in\Delta_R^+} |G_\mu(z)-G_\nu(z)| \leq e^{-m}$$
for some $m>m_0$, the Kolmogorov distance $d(\mu,\nu):=\sup_{x\in\R} |\F_\mu(x)-\F_\nu(x)|$ fulfills
$$d(\mu,\nu)\leq \Theta\frac{1}{m^\beta}.$$
\end{theorem}

\subsection{Proof of Theorem \ref{Cauchy}}

We use some results of complex analysis on the unit disc $\D:=\{z\in\C;\ |z|<1\}$. For any $f\in\O(\D)$, i.e. for any holomorphic function $f$ on $\D$ with values in $\C$, we define
$$M(f,r) := \max_{\theta\in[0,2\pi]} |f(re^{i\theta})| \qquad \text{for all $r\in[0,1)$}.$$
Moreover, for any given continuous function $\omega: [0,1) \rightarrow (0,\infty)$, we consider
$$\B(\omega) := \big\{f\in\O(\D);\ \forall r\in[0,1):\ M(f,r) \leq \omega(r)\big\}.$$
The key for the proof of the following theorem is a well-known fact due to G. H. Hardy (\cite{Hardy}), which states that the function
$$\widetilde{M}(f,\cdot):\ (-\infty,0) \rightarrow \R,\ s \mapsto \log(M(f,e^s))$$
is convex for any $f\in\O(\D)$ with $f\not\equiv0$. 

\begin{theorem}\label{disc case}
Let $\omega:\ [0,1)\rightarrow (0,\infty)$ be a continuous function satisfying
\begin{equation}\label{omega-condition}
\lim_{s\rightarrow0^+} s^\alpha\log(\omega(e^{-s})) = 0 \qquad\text{for all $\alpha>0$}.
\end{equation}
Then for all $r_0\in(0,1)$, $\beta\in(0,1)$ and $c\in(0,1)$ there exists a constant $m_0>0$ such that any function $f\in\B(\omega)$ fulfilling the condition $M(f,r_0) \leq e^{-m}$ for some $m>m_0$ satisfies
$$M\Big(f,\exp\Big(\frac{2\log(r_0)}{m^\beta}\Big)\Big) \leq \exp\big(- c m^{1-\beta}\big).$$
\end{theorem}

\begin{proof}
We define the function
$$\widetilde{\omega}:\ (-\infty,0) \rightarrow \R,\ s\mapsto \log(\omega(e^s))$$
and put $s_0 := \log(r_0)$. If $m$ is sufficiently large, we have that $s_1 := \frac{s_0}{m^\beta}$ and $s_2 := \frac{2s_0}{m^\beta}$ satisfy the condition $s_0<s_2<s_1$. Then we consider
$$l:\ \R\rightarrow\R,\ s\mapsto \frac{s-s_0}{s_1-s_0}\widetilde{\omega}(s_1) - \Big(1-\frac{s-s_0}{s_1-s_0}\Big)m.$$
Since
$$\frac{s_2-s_0}{s_1-s_0} = \frac{1-\frac{2}{m^\beta}}{1-\frac{1}{m^\beta}} \quad \stackrel{m\rightarrow\infty}{\longrightarrow}\quad 1 \qquad\text{and}\qquad m^\beta\Big(1-\frac{s_2-s_0}{s_1-s_0}\Big) = \frac{1}{1-\frac{1}{m^\beta}}\quad \stackrel{m\rightarrow\infty}{\longrightarrow}\quad 1$$
we deduce, using the equations
$$\frac{l(s_2)}{m^{1-\beta}} = \frac{1}{|s_0|^\alpha}\frac{s_2-s_0}{s_1-s_0} \Big(|s_1|^\alpha\widetilde{\omega}(s_1)\Big) - m^\beta\Big(1-\frac{s_2-s_0}{s_1-s_0}\Big)$$
and \eqref{omega-condition} with $\alpha:=\frac{1-\beta}{\beta}>0$, that $\frac{l(s_2)}{m^{1-\beta}} \rightarrow -1$ as $m\rightarrow\infty$. In particular, this implies that we can find $m_0>0$ such that $l(s_2) \leq - cm^{1-\beta}$ is fulfilled for all $m>m_0$.\\
Let $m>m_0$ be given and let $f\in\B(\omega)$ be a function satisfying $M(f,r_0) \leq e^{-m}$, where we restrict ourself to the non-trivial case $f\not\equiv 0$. Because $\widetilde{M}(f,\cdot)$ is a convex function, it follows from the inequalities
$$l(s_0) = -m \geq \log\big(M(f,e^{s_0})\big) = \widetilde{M}(f,s_0) \qquad\text{and}\qquad l(s_1) = \widetilde{\omega}(s_1) \geq \widetilde{M}(f,s_1)$$
that also the desired inequality $\widetilde{M}(f,s_2) \leq l(s_2) \leq - cm^{1-\beta}$ holds.
\end{proof}

Since we are interested in the behavior of Cauchy transforms, which are holomorphic functions living on the upper half-plane, we have to translate the above theorem. Therefore we define for every $a>0$ a biholomorphic mapping
$$\psi_a:\ \D\rightarrow\C^+,\ z\mapsto ia\frac{1+z}{1-z}.$$

\begin{theorem}\label{half-plane case}
Let $\omega: [0,1)\rightarrow (0,\infty)$ be a continuous function, which satisfies condition \eqref{omega-condition}, and let $a_0>0$ be an arbitrary constant. We consider
$$\B^\ast(\omega,a_0) := \{f\in\O(\C^+);\ \forall a\geq a_0:\ f\circ\psi_a\in\B(\omega)\}.$$
Let $\beta,c\in(0,1)$ and $R,T>0$ be given. Then we can find constants $m_0>0$ and $\eta_0>0$ such that any function $f\in\B^\ast(\omega,a_0)$, which satisfies $\sup_{z\in\Delta^+_R} |f(z)| \leq e^{-m}$ for some $m>m_0$, also fulfills
$$\max_{t\in[-T,T]} \bigg|f\bigg(i\eta_0\frac{1-\exp\big(-\frac{1}{m^\beta}\big)}{1+\exp\big(-\frac{1}{m^\beta}\big)} + t\bigg)\bigg| \leq \exp\big(-cm^{1-\beta}\big).$$
\end{theorem}

\begin{proof}
It is an easy exercise to show that the following equality holds for all $z\in\D$
$$\bigg|\psi_a(z)-ia\frac{1+|z|^2}{1-|z|^2}\bigg| = \frac{2a|z|}{1-|z|^2}$$ 
and to deduce that $\psi_a$ maps the disc $D(0,r)$ with $r>0$ to the disc $D\big(ia\frac{1+r^2}{1-r^2}, \frac{2ar}{1-r^2}\big)$.
If we put $r_0 := \exp(-\frac{1}{2}) \in (0,1)$, it is therefore possible to choose $a\geq a_0$ such that $\psi_a(D(0,r_0)) \subset \Delta_R^+$ is satisfied. Hence we have $M(f\circ\psi_a,r_0)\leq \sup_{z\in\Delta^+_R} |f(z)|$. An application of Theorem \ref{disc case} shows that, in the case where $\sup_{z\in\Delta^+_R} |f(z)| \leq e^{-m}$ holds for some $m>m_0$, we also have
$$M\Big(f\circ\psi_a,\exp\Big(\frac{2\log(r_0)}{m^\beta}\Big)\Big) \leq \exp\big(- c m^{1-\beta}\big).$$
This means that $|f|$ is bounded by $\exp\big(- c m^{1-\beta}\big)$ on the disc $\psi_a(D(0,r))$ with $r:=\exp\big(-\frac{1}{m^\beta}\big)$. For an appropriate choice of $\eta_0>0$ (an easy calculation shows that each $\eta_0>2\frac{a^2+T^2}{a}$ works), the disc $D\big(ia\frac{1+r^2}{1-r^2}, \frac{2ar}{1-r^2}\big)$ contains the whole interval $i\eta_0\frac{1-r}{1+r} + [-T,T]$ for sufficiently large values of $m$. This leads finally to the desired inequality.
\end{proof}

The validity of Theorem \ref{Cauchy} relies on the fact that we can apply this result to functions of the form $G_\mu-G_\nu$ with probability measures $\mu$ and $\nu$ on $\R$.\\
Since $\psi_a: \D\rightarrow\C^+$ gives by extension a homeomorphism $\psi_a|_{\partial\D\backslash\{1\}}: \partial\D\backslash\{1\} \rightarrow \R$, we deduce that the Cauchy transform $G_\mu$ of any probability measures $\mu$ on $\R$ transforms in the following way
$$G_\mu(\psi_a(z)) = \int_\R \frac{1}{\psi_a(z)-t} d\mu(t) = \frac{1}{2ai} \int_{\partial\D\backslash\{1\}} \frac{(1-z)(1-\xi)}{z-\xi} d\widehat{\mu}_a(\xi) \qquad\text{for all $z\in\D$},$$
where the probability measure $\widehat{\mu}_a$ is given as the push forward of $\mu$ by $(\psi_a|_{\partial\D\backslash\{1\}})^{-1}$. Thereby we used the formula
$$\frac{1}{\psi_a(z)-\psi_a(\xi)} = \frac{1}{2ai} \frac{(1-z)(1-\xi)}{z-\xi} \qquad\text{for all $\xi\in\partial\D\backslash\{1\}$ and $z\in\D$}.$$
Since $|1-z| \leq 2$, $|1-\xi| \leq 2$ and $|z-\xi| \geq 1-|z|$ holds for all $\xi\in\partial\D$ and $z\in\D$, we easily see $M(G_\mu\circ\psi_a,r) \leq \frac{2}{a} \cdot\frac{1}{1-r}$ for $0\leq r<1$.\\
If $\mu$ and $\nu$ are two probability measures on $\R$, we deduce with the above result that
$$M\big((G_\mu-G_\nu)\circ\psi_a,r\big) \leq \frac{4}{a} \cdot \frac{1}{1-r} \leq \frac{1}{1-r}$$
holds for all $0\leq r<1$ and $a\geq a_0:=4$. If we define $\omega: [0,1) \rightarrow (0,\infty)$ by $\omega(r):=\frac{1}{1-r}$, it is an easy exercise to show that $\omega$ satisfies condition \eqref{omega-condition}. Thus we get $G_\mu-G_\nu \in \B^\ast(\omega,a_0)$ and Theorem \ref{Cauchy} follows immediately from Theorem \ref{half-plane case}.

\subsection{Proof of Theorem \ref{Kolmogorov}}

At first, we recall a result of Z. D. Bai (Corollary 2.3 in \cite{Bai}), which particularly states that there are constants $c_1,c_2>0$ such that
$$d(\mu,\nu) \leq c_1 \bigg(\int^{c_2A}_{-c_2A} |G_\mu(i\eta + t) - G_\nu(i\eta + t)| dt + \frac{1}{\eta}\sup_{x\in\R} \int^{4\eta}_{-4\eta} |\F_\mu(x+t)-\F_\mu(x)| dt\bigg)$$
holds for all probability measures $\mu$ and $\nu$ with supports belonging to the interval $[-A,A]$ and for all $\eta>0$. If $\mu$ satisfies \eqref{assumption}, we apply Theorem \ref{Cauchy} in the case $T=c_2A$ and for any probability measure $\nu$ with compact support contained in $[-A,A]$, which fulfills $\sup_{z\in\Delta^+_R}|G_\nu(z)-G_\mu(z)| \leq e^{-m}$ for some $m>m_0$, we choose
$$\eta(m) = \eta_0\frac{1-\exp\big(-\frac{1}{m^\beta}\big)}{1+\exp\big(-\frac{1}{m^\beta}\big)}.$$
Hence we get
$$\int^{c_2A}_{-c_2A} |G_\mu(i\eta(m) + t) - G_\nu(i\eta(m) + t)| dt \leq 2c_2A \exp\big(-cm^{1-\beta}\big)$$
and by assumption \eqref{assumption}
$$\sup_{x\in\R} \int^{4\eta(m)}_{-4\eta(m)} |\F_\mu(x+t)-\F_\mu(x)| dt \leq \rho \int^{4\eta(m)}_{-4\eta(m)} |t| dt = 16\rho\eta(m)^2.$$
Since $\lim_{m\rightarrow\infty} m^\beta\eta(m) = \frac{1}{2}\eta_0$ and $\lim_{m\rightarrow\infty} m^\beta \exp(-cm^{1-\beta}) = 0$, it follows that
$$d(\mu,\nu) \leq \Big(2c_1c_2A m^\beta\exp\big(-cm^{1-\beta}\big) + 16c_1\rho m^\beta\eta(m)\Big) \frac{1}{m^\beta} \leq \Theta\frac{1}{m^\beta}$$
holds with an appropriate choice of $\Theta>0$ (for instance, each $\Theta > 8c_1\rho\eta_0$ works) for all sufficiently large values of $m$.

\end{document}